\documentclass[letterpaper,conference]{IEEEtran}
\IEEEoverridecommandlockouts
\usepackage{amsmath, bm, amssymb, algpseudocode, algorithm, mathtools, svg,ifluatex, pgfplots, amsthm, caption, subcaption, tikz}
\usepackage{cite, bbm, accents, url}
\usepackage{graphicx}
\usepackage{textcomp}
\usepackage{xcolor}
\usepackage{stfloats,cuted}
\usepackage[bookmarks=false]{hyperref}
\usepackage[acronym,shortcuts]{glossaries}
\usepackage[letterpaper, bottom=1.05in, top=0.75in, left=0.63in, right=0.63in]{geometry}
\usetikzlibrary{graphs, positioning, quotes, shapes.geometric, arrows}

\pgfplotsset{compat=1.15}

\newacronym{ba}{BA}{Blahut-Arimoto}
\newacronym{cc}{C-C}{capacity-cost}
\newacronym{wgd}{WGD}{Wasserstein gradient descent}
\newacronym{kld}{KLD}{Kullback-Leibler divergence}
\newacronym{rd}{R-D}{rate-distortion}
\newacronym{ac}{a.c.}{absolutely continuous}
\newacronym{rnd}{RND}{Radon-Nikodym derivative}
\newacronym{pdf}{PDF}{probability density function}
\newacronym{jko}{JKO}{Jordan-Kinderlehrer-Otto}
\newacronym{mci}{MCI}{Monte-Carlo integration}
\newacronym{is}{IS}{importance sampling}
\newacronym{csir}{CSIR}{channel state information at receiver}

\newcommand*{\br}[1]{\left ( #1 \right )}

\newcommand*{\brrr}[1]{\left \{ #1 \right \}}

\newcommand{\RR}{\mathbb{R}}

\newcommand{\calX}{\mathcal{X}}
\newcommand{\calY}{\mathcal{Y}}
\newcommand{\calP}{\mathcal{P}}
\newcommand{\calW}{\mathcal{W}}
\newcommand{\calN}{\mathcal{N}}

\newcommand{\kx}{\kappa_{x}}

\newcommand{\dx}{\mathrm{d}x}
\newcommand{\dy}{\mathrm{d}y}
\newcommand{\dmu}{\mathrm{d}\mu}
\newcommand{\dnu}{\mathrm{d}\nu}
\newcommand{\dkx}{\mathrm{d}\kx}

\newcommand*{\st}{\mathrm{s. \, t. \,}}

\DeclareMathOperator*{\argmax}{arg\,max}

\DeclareMathOperator*{\arginf}{arg\,inf}

\newtheorem{theorem}{Theorem}

\theoremstyle{definition}

\tikzset{%
    smallblock/.style={draw, fill=white, minimum height=1em, minimum width=1em},
    block-common/.style={draw, fill=white, minimum height=1.5em, minimum width=4em},
    block/.style={rectangle, block-common, align=center},
    txtblock/.style={block, align=center, minimum height=4em},
    bigblock/.style={block, minimum height=8em},
    txtbigblock/.style={bigblock, align=center},
    input/.style={inner sep=1pt},       
    output/.style={inner sep=1pt},      
    sum/.style = {draw, fill=white, circle, minimum size=1.1em, inner sep=0pt,
      font={\small$+$}},
    prod/.style = {draw, fill=white, circle, minimum size=1.1em, inner sep=0pt,
      font={\normalsize$\times$}},
    pinstyle/.style = {pin edge={to-,thin,black}}
}

\begin{document}
\title{Computing Capacity-Cost Functions for Continuous Channels in Wasserstein Space
\thanks{The authors were supported in part by the
    German Federal Ministry of Education and Research (BMBF)
    in the programme “Souver\"an. Digital. Vernetzt.”
    within the research hub 6G-life under Grant 16KISK002,
    and also by the Bavarian Ministry of Economic Affairs,
    Regional Development and Energy within the project 6G Future Lab Bavaria.
    U. M{\"o}nich and H. Boche were also supported by the BMBF within the project "Post Shannon Communication - NewCom" under Grant 16KIS1003K.
    X. Li, V. Andrei, and H. Boche also acknowledge the financial support the financial support by the BMBF Projects QD-CamNetz, Grant
    16KISQ077, QuaPhySI, Grant 16KIS1598K, and QUIET, Grant 16KISQ093.
    }
}

\author{\IEEEauthorblockN{Xinyang Li\IEEEauthorrefmark{1},
Vlad C. Andrei\IEEEauthorrefmark{1}, Ullrich J. M{\"o}nich\IEEEauthorrefmark{1}, Fan Liu\IEEEauthorrefmark{2} and
Holger Boche\IEEEauthorrefmark{1}\IEEEauthorrefmark{3}}
\IEEEauthorblockA{\IEEEauthorrefmark{1}Chair of Theoretical Information Technology, Technical University of Munich, Munich, Germany\\
\IEEEauthorrefmark{2}School of Information Science and Engineering, Southeast University, Nanjing 210096, China\\
\IEEEauthorrefmark{1}BMBF Research Hub 6G-life,
\IEEEauthorrefmark{3}Munich Center for Quantum Science and Technology,
\IEEEauthorrefmark{3}Munich Quantum Valley\\
Email: \{xinyang.li, vlad.andrei, moenich\}@tum.de, f.liu@ieee.org, boche@tum.de}
}

\maketitle

\allowdisplaybreaks
\glsdisablehyper

\begin{abstract}
This paper investigates the problem of computing \ac{cc} functions for continuous channels. Motivated by the \ac{kld} proximal reformulation of the classical \ac{ba} algorithm, the Wasserstein distance is introduced to the proximal term for the continuous case, resulting in an iterative algorithm related to
the Wasserstein gradient descent. Practical implementation involves moving particles along the negative gradient direction of the objective function's first variation in the Wasserstein space and approximating integrals by the \ac{is} technique. Such formulation is also applied to the \ac{rd} function for continuous source spaces and thus provides a unified computation framework for both problems. 
\end{abstract}


\glsresetall

\section{Introduction}\label{sec:intro}
As the maximum rate for reliable message transmission, channel capacity plays a crucial role in determining the performance limits of communication systems. It is given by the maximum mutual information between the channel input and output, optimized over all possible input distribution\cite{gallager1968information}. If an input cost constraint is further imposed, e.g., input power, the relationship between channel capacity and input cost is referred to as \ac{cc} function.

For discrete memoryless channels, where the input and output alphabets are finite, the \ac{ba} algorithm\cite{blahut1972computation, arimoto1972algorithm} is capable of computing the \ac{cc} function by iteratively updating the input distribution and is shown to converge to the global optimum. In \cite{matz2004information, naja2009geometrical}, the computation problem is reformulated as a \ac{kld} proximal point problem, providing an alternative point of view and a way to accelerate convergence.

It has been shown in the current publications\cite{boche2023algorithmic,lee2024computability} that the \ac{ba} algorithms or the approach from\cite{matz2004information, naja2009geometrical} cannot be used to approximately calculate optimal input distributions. The results\cite{boche2023algorithmic,lee2024computability} also indicate that calculating the \ac{cc} function for continuous channels is very challenging.
Since the input distribution is described by a \ac{pdf}, or more generally, probability measure, standard optimization tools in finite-dimensional Euclidean spaces are no longer applicable. Existing works include \cite{chang1988calculating, dauwels2005numerical}, both of which approximate the channel \ac{pdf} by a finite partition of the output space and sequentially move a particle set in the input space to increase the inner integral of the \ac{cc} function. However, \cite{chang1988calculating} assumes that all the local optima of the inner integral at each iteration can be found to guarantee the convergence, which is often too idealistic. \cite{dauwels2005numerical} performs the steepest gradient descent and coincides with our proposed method to some extent but doesn't provide a convergence analysis and detailed interpretation. Moreover, numerical integration with fixed finite partitions may result in an inaccurate approximation or extremely complex computation for high-dimensional cases.

In this work, we adopt recent results on Wasserstein gradient flow\cite{ambrosio2008gradient,santambrogio2017euclidean}, originated from the field of optimal transport \cite{villani2009optimal}, to reformulate the computation problem in the Wasserstein space by following the ideas in \cite{matz2004information,ambrosio2008gradient}. We develop a gradient descent type method in the Wasserstein space, alternately updating the input and output probability measures, whose local convergence is shown to be guaranteed under certain conditions. Numerically, the transport of probability measures is implemented by applying the corresponding transformation on particles, and the \ac{is} technique is performed to estimate the involved integrals. Simulation results\footnote{Code is available at {\url{https://github.com/xinyanglii/continuous_channel}}.} are provided for both MIMO-AWGN and fading channels, demonstrating a convincing performance of the proposed method.

\section{Background and Notations}
\subsection{Blahut-Arimoto Algorithm}

Let $P_{Y|X}$ be a channel transition matrix from a finite input space $\calX$ to a finite output space $\calY$. The set of all possible distributions on $\calX$ is a probability simplex\cite{boyd2004convex}
\begin{equation}
    \calP(\calX) \triangleq \brrr{P_X \,\bigg| \, \forall x \in \calX, P_X(x) \ge 0, \sum_{x\in \calX} P_X(x) = 1}.
    \label{eq:simplex}
\end{equation}
The \ac{ba} algorithm is proposed in\cite{blahut1972computation, arimoto1972algorithm} to solve the optimization problem arising from computing the corresponding \ac{cc} function:
\begin{equation}
   C(B) \triangleq \max_{P_X\in \calP(\calX)} I(X,Y),\ \st \sum_{x\in \mathcal{X}} P_{X}(x) b(x) \le B,
   \label{eq:discrete_cc}
\end{equation}
where $I(X,Y)$ is the mutual information between channel input $X$ and output $Y$, $b: \calX \to [0, \infty)$ is a cost function, imposing an additional input constraint. The \ac{ba} algorithm formulates \eqref{eq:discrete_cc} as a double maximization problem over $P_X$ and $P_{X|Y}$, 
and iteratively updates them until convergence.

Another formulation of problem \eqref{eq:discrete_cc} is provided by\cite{matz2004information, naja2009geometrical} from the perspective of proximal point methods\cite{rockafellar1976monotone}. Specifically, the update rule on $P_X$ at step $k$ is given by
\begin{equation}
\begin{split}
    &P^{(k)}_X =  \argmax_{P_X \in \calP(\calX)} I(X,Y) + \lambda \sum_{x\in \mathcal{X}} P_{X}(x) b(x) \\
    & \hspace{25mm}  - \frac{D(P_X \| P_X^{(k-1)})}{ \tau_k} + D(P_Y \| P_Y^{(k-1)}),
    \label{eq:kld_reform}
\end{split}
\end{equation}
where $D(\cdot \| \cdot)$ is the \ac{kld}, $\lambda$ is the Lagrangian multiplier, and $\tau_k$ is the proximal step size at step $k$.
The solution to \eqref{eq:kld_reform} is 
\begin{equation*}
    P^{(k)}_X(x)= \frac{P_X^{(k-1)}(x) \exp \br{ \tau_k  D_\lambda^{(k)}(x)}}{\sum_{x'}P_X^{(k-1)}(x') \exp\br{ \tau_k  D_\lambda^{(k)}(x')}}
\end{equation*}
for all $x\in \calX$, where
\begin{equation*}
    D_\lambda^{(k)}(x) \triangleq \sum_{y} P_{Y|X}(y|x) \log\frac{P_{Y|X}(y|x)}{P_Y^{(k-1)}(y)} - \lambda b(x).
\end{equation*}
The update rule becomes the classical \ac{ba} step when $\tau_k = 1$,
and thus not only interprets the \ac{ba} algorithm from another viewpoint but also enables accelerating the convergence by choosing $\tau_k$ adaptively at each step.

\subsection{Wasserstein Space}

For the continuous case, the probability space can not be considered as a finite vector subspace as in \eqref{eq:simplex}. Necessary concepts from measure theory and optimal transport are thus briefly reviewed in this subsection. 

For the $n$-dimensional Euclidean space $\RR^n$, the set of all Borel probability measures is denoted by $\calP(\RR^n)$ and its subset with finite second moment by
\begin{equation*}
    \calP_2(\RR^n) \triangleq \brrr{\mu \in  \calP(\RR^n) \Big| \int_{\RR^n} \| x \|_2^2\  \mu(\dx) < + \infty }.
\end{equation*}
For two measures $\mu, \nu \in \calP(\RR^n)$, we write $\mu \ll \nu$ if $\mu$ is \ac{ac} with respect to $\nu$, and as a result, the \ac{rnd} $\frac{\dmu}{\dnu}: \RR^n \to [0, +\infty)$ exists\cite{billingsley2012probability}. If $\mu$ is \ac{ac} with respect to the Lebesgue measure (\ac{ac} for short), the \ac{rnd} $\frac{\dmu}{\dx}$ normally refers to the \ac{pdf}. 

Given two Borel spaces $(\RR^n, \mathcal{B}(\RR^n))$ and $(\RR^m, \mathcal{B}(\RR^m))$, The Markov kernel\cite{ccinlar2011probability} is a mapping $\kappa: \RR^n  \times \mathcal{B}(\RR^m)  \to [0, 1]$ such that for every Borel set $B \in \mathcal{B}(\RR^m)$, the mapping $x \mapsto \kx (B)$ is a $\mathcal{B}(\RR^n)$-measurable function, and $\kx$ is a Borel probability measure on $\RR^m$ for all $x \in \RR^n$. Let $T: \RR^n \to \RR^n$ be a $\mathcal{B}(\RR^n)$-measurable function, the pushforward measure\cite{ambrosio2008gradient} of $\mu$ under $T$ is $T_\#\mu$ such that
\begin{equation}
    T_\#\mu(B) = \mu(T^{-1}(B)), \quad \forall B \in \mathcal{B}(\RR^n).
\end{equation}

The 2-Wasserstein distance\cite{ambrosio2008gradient} between two measures $\mu, \nu \in \calP_2(\RR^n)$ is defined as
\begin{equation}\label{eq:w2prob}
     W_2(\mu, \nu) \triangleq \left (\min_{\pi \in \Pi(\mu, \nu)}\int_{\mathbb{R}^n\times \mathbb{R}^n} \|x - y \|_2^2\  \pi(\dx,\dy) \right)^{\frac{1}{2}},
\end{equation}
where $\Pi(\mu, \nu)$ is the set of all couplings between $\mu$ and $\nu$, i.e., all the product probability measures on $\RR^n \times \RR^n$ such that the first and second marginals are $\mu$ and $\nu$ respectively. The optimal coupling $\pi^*$ achieving the minimum in \eqref{eq:w2prob} is called the optimal transport plan. In the following, we call $W_2$ the Wasserstein distance for convenience. The optimization problem involved in $W_2$ can be recast as the dual form
\begin{equation}
    \sup_{\varphi \in L^1(\mu)} \int \varphi(x) \mu(\dx) + \int \varphi^c(x) \nu(\dx)
    \label{eq:wassersteindual}
\end{equation}
where $\varphi^c$ is the c-transform of $\varphi$:
\begin{equation}
    \varphi^c(y) = \sup_{x \in \RR^n} \varphi(x) - \| x- y\|_2^2, \quad \forall y \in \RR^n,
\end{equation}
and the resulting optimizer $\varphi$ to \eqref{eq:wassersteindual} is called the Kantorovich potential.
Due to Brenier's theorem\cite{brenier1991polar, santambrogio2017euclidean}, if $\mu$ is \ac{ac}, the optimal transport plan $\pi^*$ is uniquely determined by $\pi^* = (\mathrm{Id}, T_\mu^\nu)_\# \mu$ with the identity map $\mathrm{Id}$ on $\RR^n$, the optimal transport map $T_\mu^\nu(x) = x - \nabla \varphi(x)$ and the Kantorovich potential $\varphi$.

It is shown that $W_2$ satisfies the properties of a metric, and thus $\calW_2 \triangleq (\calP_2(\RR^n), W_2)$ forms a metric space called the Wasserstein space\cite{villani2009optimal}. It also turns out that $\calW_2$ is equipped with a differential structure: for a functional $F: \calP_2(\RR^n) \to \RR$, its first variation\cite{santambrogio2015optimal}, if exists, is given by $\frac{\delta F}{\delta \mu}: \RR^n \to \RR$ such that for all $\chi$ with $\mu + \epsilon \chi \in \calP_2(\RR^n) $
\begin{equation*}
    \lim_{\epsilon \to 0^+} \frac{F(\mu + \epsilon \chi) - F(\mu)}{\epsilon} = \int_{\RR^n} \frac{\delta F}{\delta \mu} (x)\  \chi(\dx).
\end{equation*}
A simple example is when $F(\mu) = \int V(x) \mu(\dx)$ for a function $V: \RR^n \to \RR$, then $\frac{\delta F}{\delta \mu} (x) = V(x)$. Moreover, the first variation of the function $\mu \mapsto W_2^2(\mu, \nu)$ is given by $\frac{\delta W_2^2(\cdot, \nu)}{\delta \mu} (x) = \varphi(x)$ up to additive constants with the Kantorovich potential $\varphi$\cite{santambrogio2015optimal}.

\section{Proposed Method}
\subsection{Problem Formulation}

Let $\RR^n$ and $\RR^m$ be the input and output space. Throughout this paper, the channel is assumed to be characterized by a Markov kernel $\kappa$ with $\kx$ \ac{ac} for every $x \in \RR^n$ and $\kx(B)$ is differentiable on $\RR^n$ for every $B \in \mathcal{B}(\RR^m)$. Hence, the conditional \ac{pdf} of the channel is well defined and given by $p_{Y|X}(y|x) \triangleq \frac{\dkx}{\dy}(y)$. For an input probability measure $\mu \in \calP_2(\RR^n)$, the corresponding output  measure is $\nu = \int_{\RR^n} \kx \mu(\dx)$, which is obviously \ac{ac} and thus also possesses a \ac{pdf} $p_Y(y) \triangleq \frac{\dnu}{\dy}(y) = \int_{\RR^n} p_{Y|X}(y|x)\ \mu(\dx)$. For ease of analysis, the notations of \ac{pdf} and \ac{rnd} are used interchangeably in the following, depending on the context.

Defining a differentiable cost function on the input space $b: \RR^n \to [0, +\infty)$, the \ac{cc} function is given as\cite{gallager1968information}
\begin{equation*}
    C(B) = \sup_{\mu\in \calP_2(\RR^n)} I(X, Y), \quad \st \int_{\RR^n} b(x)\ \mu(\dx) \le B,
\end{equation*}
where the mutual information for the continuous case is
\begin{equation*}
    I(X,Y) = \int_{\RR^n}\mu(\dx) \int_{\RR^m} \kx(\dy) \log\frac{\dkx}{\dnu}(y)
\end{equation*}
if $\kx \ll \nu$ for all $x$ and infinite otherwise. Similarly to the discrete case, the corresponding Lagrangian function to be minimized reads as
\begin{equation}
    L_\lambda(\mu) \triangleq \int_{\RR^n}\mu(\dx) \br{\lambda b(x) - \int_{\RR^m} \kx(\dy) \log\frac{\dkx}{\dnu}(y)}
    \label{eq:lagrang}
\end{equation}
with $\lambda \ge 0$ representing the slope of the curve $C(B)$.

\subsection{Alternating Updates with Wasserstein Gradient Descent}

Motivated by the proximal point reformulation in \eqref{eq:kld_reform}, the proposed method replaces the \ac{kld} penalty term $D(P_X\| P_X^{(k-1)})$ by the Wasserstein distance between $\mu$ and $\mu^{(k-1)}$ at step $k$, resulting in a sequence of problems
\begin{equation}
\begin{split}
    &\mu^{(k)} = \arginf_{\mu} L_{\lambda} (\mu) + \frac{W_2^2(\mu, \mu^{(k-1)})}{\tau_k} - D(\nu \| \nu^{(k-1)}) \\
    &= \arginf_{\mu} \int_{\RR^n}\mu(\dx) \br{\lambda b(x) - \int_{\RR^m} \kx(\dy) \log\frac{\dkx}{\dnu^{(k-1)}}(y)}\\
    &\hspace{10mm} + \frac{W_2^2(\mu, \mu^{(k-1)})}{\tau_k}, 
\label{eq:wasser_prox}
\end{split}
\end{equation}
where $\nu^{(k-1)} = \int_{\RR^n} \kx\  \mu^{(k-1)} (\dx)$.
The optimality condition for the $k$-th step problem is that the first variation of its objective function with respect to $\mu$ is equal to a constant value\cite{ambrosio2008gradient}. It is also noted that the first part of the objective function takes the form of $\int V_{\lambda}^{(k)}(x) \mu(\dx)$ with
\begin{equation}
\begin{split}
    V_{\lambda}^{(k)}(x) &\triangleq \lambda b(x) - \int_{\RR^m} \kx(\dy) \log\frac{\dkx}{\dnu^{(k-1)}}(y)\\
    &=  \lambda b(x) - \int_{\RR^m} p_{Y|X}(y|x) \log\frac{ p_{Y|X}(y|x)}{p_Y^{(k-1)}(y)}\  \dy.
\end{split}
    \label{eq:Vx}
\end{equation}
This leads to 
\begin{equation*}
    V_{\lambda}^{(k)}(x) + \frac{\varphi^{(k)}(x)}{\tau_k} = \mathrm{const},
\end{equation*}
with the Kantorovich potential $\varphi^{(k)}(x)$ associated with $W_2(\mu, \mu^{(k-1)})$,
and thus $\nabla  V_{\lambda}^{(k)}(x) + \frac{\nabla \varphi^{(k)}(x)}{\tau_k} = 0$.
Leveraging Brenier's theorem, the optimal transport plan from $\mu^{(k)}$ to $\mu^{(k-1)}$ is obtained as
\begin{equation}
    T_{\mu^{(k)}}^{\mu^{(k-1)}}(x) \triangleq x - \nabla \varphi^{(k)}(x) = x + \tau_k \nabla  V_{\lambda}^{(k)}(x),
    \label{eq:muk_to_muk-1}
\end{equation}
and for sufficiently small $\tau_k$, the optimal transport plan from $\mu^{(k-1)}$ to $\mu^{(k)}$ is approximated as
\begin{equation}
    T^{(k)}(x) \triangleq x - \tau_k \nabla  V_{\lambda}^{(k)}(x).
    \label{eq:update_rule_x}
\end{equation}
Therefore, the update rule for $\mu^{(k)}$ becomes
\begin{equation}
    \mu^{(k)} = (T_{\mu^{(k)}}^{\mu^{(k-1)}})^{-1}_\# \mu^{(k-1)} \approx T^{(k)}_\# \mu^{(k-1)}.
    \label{eq:update_rule_mu}
\end{equation}

It is worth mentioning that the approximated solution \eqref{eq:update_rule_mu} to \eqref{eq:wasser_prox} can be considered analogously to the explicit Euler scheme of the gradient flow in Euclidean spaces\cite{santambrogio2017euclidean}, 
which is computationally efficient but may suffer from potential instability problems. An implicit scheme is also possible to derive but is beyond the scope of this paper. Nevertheless, simulation results show that the explicit scheme can still yield remarkable performance.

\subsection{Numerical Implementation}

Given $\mu \in \calP_2(\RR^n)$ and a set of $N$ particles $\{x_i\}_{i=1}^N$ sampled from $\mu$, the empirical probability measure $\mu_N \triangleq \frac{1}{N} \sum_{i=1}^N \delta_{x_i}$,
represented by averaging Dirac measures at sample points, is usually used to approximate $\mu$, where $\delta_{x_i}$ is the Dirac measure at $x_i$.
It also shows that the pushforward measure $T_\# \mu_N$ for any measurable function $T$ is again an empirical measure consisting of a set of particles $\{y_i\}_{i=1}^N$ with $y_i = T(x_i)$ for all $i$.
Hence, starting from a initial empirical measure $\mu_N^{(0)} = \frac{1}{N} \sum_{i=1}^N \delta_{x_i^{(0)}}$, the update rule \eqref{eq:update_rule_mu} results in a sequence of particle sets $\{x_i^{(k)}\}_{i=1}^N$ for $k=1,2,...$ with
\begin{equation}
    x_i^{(k)} = T^{(k)}(x_i^{(k-1)}) = x_i^{(k-1)} - \tau_k \nabla V_{\lambda}^{(k)}(x_i^{(k-1)}), \  \forall i,
    \label{eq:update_particle}
\end{equation}
where $p_Y^{(k-1)}(y)$ in $V_{\lambda}^{(k)}(x_i^{(k-1)})$ is computed based on $\mu_N^{(k-1)}$ as
\begin{equation}
\begin{split}
    p_Y^{(k-1)}(y) &= \int_{\RR^n} \frac{1}{N} \sum_{i=1}^N p_{Y|X}(y|x) \delta_{x^{(k-1)}_i}(\dx)\\
    &= \frac{1}{N} \sum_{i=1}^N  p_{Y|X}(y|x^{(k-1)}_i).
    \label{eq:py_approx}
\end{split}
\end{equation}

Besides representing $\mu^{(k)}$ by a particle set, computing $V_{\lambda}^{(k)}(x)$ and $\nabla V_{\lambda}^{(k)}(x)$ also requires numerical approaches because the involved integration in \eqref{eq:Vx} typically has no analytical expression. The finite partition-based method is used in\cite{dauwels2005numerical}, which shows high accuracy but leads to exponentially increasing computational complexity in dimension. To this end, the \ac{is} technique is adopted here as it can balance complexity and accuracy by choosing different numbers of samples.

To compute an integral $I = \int f(a) \mathrm{d} a$, the basic idea of \ac{is} is to draw a set of $N_s$ samples $\{a_i\}_{i=1}^{N_s}$ from a pre-defined distribution $q(a)$, called the importance distribution\cite{candy2016bayesian}, and then approximate the integration $I$ by 
\begin{equation*}
    I \approx \hat{I} = \int  \frac{1}{N_s} \sum_{i=1}^{N_s} \frac{f(a)}{q(a)} \delta_{a_i}(\mathrm{d}a) = \frac{1}{N_s} \sum_{i=1}^{N_s} \frac{f(a_i)}{q(a_i)}.
\end{equation*}
The similar procedure is applied to $V_{\lambda}^{(k)}(x)$ and $\nabla V_{\lambda}^{(k)}(x)$, resulting in 
\begin{equation}
    \hat{V}_{\lambda}^{(k)}(x) =\lambda b(x) - \frac{1}{N_s} \sum_{i=1}^{N_s} \frac{p_{Y|X}(y_i|x) \log\frac{p_{Y|X}(y_i|x)}{p_Y^{(k-1)}(y_i)}}{q(y_i|x)}
    \label{eq:vx_approx}
\end{equation}
and
\begin{equation}
\begin{split}
    &\widehat{\nabla V}_{\lambda}^{(k)}(x)= \\
    &\quad \lambda \nabla b(x) - \frac{1}{N_s} \sum_{i=1}^{N_s} \frac{ \nabla \br{p_{Y|X}(y_i|x) \log\frac{p_{Y|X}(y_i|x)}{p_Y^{(k-1)}(y_i)}}}{q(y_i|x)},
    \label{eq:dvx_approx}
\end{split}
\end{equation}
where it is emphasized that the importance distribution $q(y|x)$ may also depend on $x$.

In the above analysis, the Lagrangian multiplier $\lambda$ is fixed throughout. As a result, the final cost value is unknown in advance. In order to obtain the channel capacity at a particular cost, methods such as bisection search are commonly adopted but require running the whole algorithm multiple times at different values of $\lambda$. To deal with this problem, the dual ascent method\cite{boyd2011distributed} is considered, which additionally updates $\lambda$ during running. 

Putting steps together, the pseudocodes of the proposed numerical method are summarized in Algorithm \ref{alg:compute}. The outputs contain not only the resulting particles representing the optimal input distribution for the considered channel but also the approximated objective values, rate, and cost. In practice, the computation of gradient involved in \eqref{eq:dvx_approx} can be conducted by various automatic differentiation frameworks like PyTorch, and the step size $\tau_k$ may also be adjusted using adaptive optimizers such as Adam\cite{kingma2014adam}.

\begin{algorithm}[t]
\caption{Numerical computation of \ac{cc} function}
\begin{algorithmic}
    \State \textbf{Inputs:} 
    \State Channel distribution $p_{Y|X}$; cost function $b$; initial measure $\mu^{(0)}$; number of particles $N$; sequence of importance distributions $q^{(0)}(y|x), q^{(1)}(y|x),...$; number of \ac{is} samples $N_s$; sequence of step sizes $\tau_1, \tau_2,...$; initial dual multiplier $\lambda_0$
    \State \textit{Optional}: cost upper bound $B$, sequence of step sizes for dual ascent update $\alpha_0, \alpha_1,...$
    \State \textbf{Initialize:}
    \State Sample a particle set $\{x_i^{(0)}\}_{i=1}^N$ from $\mu^{(0)}$;
    \State $k \gets 0$
    \While{convergence condition not met}
    \State Sample a particle set $\{y_{j_i}^{(k)}\}_{j_i=1}^{N_s}$ from $q^{(k)}(y|x_i^{(k)})\  \forall i$
    \State Compute $p_Y^{(k)}(y_{j_i}^{(k)})$ via \eqref{eq:py_approx}
    \State Compute $\hat{V}_{\lambda_k}^{(k+1)}(x_i^{(k)})$, $\widehat{\nabla V}_{\lambda_k}^{(k+1)}(x_i^{(k)})$ via \eqref{eq:vx_approx}, \eqref{eq:dvx_approx}
    \State $\hat{L}_{\lambda_k} \gets \frac{1}{N} \sum_{i=1}^N \hat{V}_{\lambda_k}^{(k+1)}(x_i^{(k)})$ \Comment{see \eqref{eq:lagrang}}
    \State $\hat{B}^{(k)} \gets \frac{1}{N} \sum_{i=1}^N b(x_i^{(k)})$
    \State $\hat{R}^{(k)} \gets  \lambda_k \hat{B}^{(k)}- \hat{L}_{\lambda_k}$
    \State $x_i^{(k+1)} \gets x_i^{(k)} - \tau_{k+1} \widehat{\nabla V}_{\lambda_k}^{(k+1)}(x_i^{(k)})$ \Comment{see \eqref{eq:update_rule_x}, \eqref{eq:update_particle}}
    \If{dual ascent update required}
    \State $\lambda_{k+1} \gets \max\{0,\lambda_k + \alpha_k (\hat{B}^{(k)} - B)\}$
    \Else
    \State $\lambda_{k+1} \gets \lambda_k$
    \EndIf
    \State $k \gets k+1$
    \EndWhile
    \State Evaluate $\hat{L}_{\lambda_{k}}$, $\hat{R}^{(k)}$, $\hat{B}^{(k)}$ on $\{x_i^{(k)}\}_{i=1}^N$
    \State \textbf{Return:}
    \State $\{x_i^{(k)}\}_{i=1}^N$, $\hat{L}_{\lambda_{k}}$, $\hat{R}^{(k)}$, $\hat{B}^{(k)}$
\end{algorithmic}\label{alg:compute}
\end{algorithm}

\section{Rate-Distortion Function}
In fact, computing \ac{rd} functions for continuous source space with Wasserstein gradient descent is already studied in \cite{yang2024estimating}. 
Nonetheless, this section gives a brief review by interpreting the \ac{rd} problem within the proposed proximal point structure, aiming at providing a unified computation framework for both \ac{cc} and \ac{rd} problems.

Let $\RR^n$ be the space of source and reconstruction. Given a source probability measure $\mu \in \calP_2(\RR^n)$ and a distortion function $d: \RR^n \times \RR^n \to [0, +\infty)$, the original \ac{rd} problem is formulated as
\begin{equation}
\begin{split}
    R(D) &\triangleq \inf_{\pi \in \Pi(\mu,)}  I(X,Y) \\
    &\st \int_{\RR^n\times \RR^n} d(x,y)\ \pi(\dx, \dy) \le D .
    \label{eq:rd_prob}
\end{split}
\end{equation} 
If restricting $\pi$ to the form of $\pi(A\times B) = \int_A \kx(B) \mu(\dx) $ for all $A, B \in \mathcal{B}(\RR^n)$ and some Markov kernels $\kappa$, the problem \eqref{eq:rd_prob} is then equivalent to $\inf_\nu G_\lambda(\nu)$ with
\begin{equation*}
    G_\lambda(\nu) \triangleq - \int_{\RR^n} \mu(\dx) \log \int_{\RR^n} \nu(\dy) \exp(-\lambda d(x,y)),
\end{equation*}
where $\lambda$ is the Lagrangian multiplier, and $\nu$ is the reconstruction measure. Similarly to \eqref{eq:wasser_prox}, the update rule for $\nu^{(k)}$ based on the Wasserstein proximal point method at step $k$ is
\begin{align*}
    \nu^{(k)} &= \arginf_{\nu} G_\lambda(\nu) + \frac{W_2^2(\nu, \nu^{(k-1)})}{\tau_k}\\
    & = \br{\mathrm{Id} + \tau_k \nabla \frac{\delta G_\lambda}{\delta \nu}(\nu^{(k)})}_{\#}^{-1}\nu^{(k-1)}\\
    & \approx \br{\mathrm{Id} - \tau_k \nabla \frac{\delta G_\lambda}{\delta \nu}(\nu^{(k-1)})}_{\#}\nu^{(k-1)},
\end{align*}
in which the second line follows the optimality condition, and the last line is due to the explicit approximation scheme. The first variation of $G_\lambda$ is derived in \cite{yang2024estimating} and given by
\begin{equation*}
    \frac{\delta G_\lambda}{\delta \nu}(\nu) = - \int_{\RR^n} \frac{\exp(-\lambda d(x, \cdot))}{\int_{\RR^n} \exp(-\lambda d(x, \tilde{y})) \nu(\mathrm{d}\tilde{y})} \mu(\dx).
\end{equation*}

Again, the \ac{is} technique can be used to numerically compute $\frac{\delta G_\lambda}{\delta \nu}(\nu)$ as well as its gradient. Thus, the algorithm in \cite{yang2024estimating} is recovered, demonstrating the generalizability of the proposed computation method.

\section{Convergence Analysis}
Like the convergence behavior of the Euclidean gradient flow\cite{santambrogio2017euclidean}, it shows that the Wasserstein proximal point method \eqref{eq:wasser_prox} and its explicit scheme \eqref{eq:update_rule_x} can also yield a ``stationary point" in $\calW_2$ under certain conditions. 

It is assumed in the following $\kx \ll \nu$ for all $x$ such that the mutual information is finite. This can be satisfied when, for instance, $\kx(B) > 0$ for all $B\in \mathcal{B}(\RR^m)$. The following theorem states a sufficient condition for \eqref{eq:wasser_prox} producing a converging sequence of $\mu^{(k)}$.
\begin{theorem}\label{th:stationary}
    Let $\{\mu^{(k)}\}_{k=1}^\infty$ be the sequence of probability measures generated by \eqref{eq:wasser_prox}, and suppose for all $k$ the step sizes satisfy
    \begin{equation}
       0 < \tau_k < \inf_{\mu \neq \mu^{(k-1)}} \frac{W_2^2 (\mu, \mu^{(k-1)})}{D(\nu \| \nu^{(k-1)})}.
    \end{equation}
    Then $\mu^{(k)}$ converges to a measure $\mu^*$ in the sense that 
    \begin{equation}
    \lim_{k\to \infty} W_2(\mu^{(k)}, \mu^*) = 0,
    \end{equation}
    such that
    \begin{equation}
       \lim_{k\to \infty} \int_{\RR^n} \| \nabla V_\lambda^{(k)} (x) \|_2^2 \ \mu^*(\dx) = 0.
       \label{eq:wasser_grad_cvg}
    \end{equation}
\end{theorem}
\begin{proof}
    The condition on $\tau_k$ guarantees the monotone decreasing of the sequence $L_{\lambda}(\mu^{(k)})$, which thus converges to a limit point because $L_{\lambda}$ is bounded. This in turn implies $W_2^2 (\mu^{(k)}, \mu^{(k-1)}) \to 0$ because the choice of $\tau_k$ guarantees that the proximal term is strictly positive unless $\mu^{(k)} =\mu^{(k-1)}$. Thus, $\mu^{(k)}$ converges to a limit point, denoted as $\mu^*\in \calP_2(\RR^n)$, due to the completeness of $\calW_2$\cite{villani2009optimal}. Next note that from \eqref{eq:muk_to_muk-1},
    \begin{equation}
        W_2^2(\mu^{(k)}, \mu^{(k-1)}) = \tau_k^2 \int\|\nabla V_\lambda^{(k)} (x)\|_2^2\  \mu^{(k-1)}(\dx).
    \end{equation}
    Taking the limit as $k\to \infty$, the right-hand side tends to zero. Furthermore, because $\tau_k > 0$ for all $k$ and the convergence in $\calW_2$ is equivalent to the weak convergence of measures\cite{ambrosio2008gradient}, 
    the second part \eqref{eq:wasser_grad_cvg} is proven.
%
\end{proof}

The next theorem shows that the proposed explicit scheme \eqref{eq:update_rule_x} can also produce a converging sequence of measures but with different requirements on the step sizes. The proof follows similar arguments as \cite{yang2024estimating} and thus is skipped here.

\begin{theorem}
    Let $\{\mu^{(k)}\}_{k=1}^\infty$ be the sequence of probability measures generated by the explicit approximation step \eqref{eq:update_rule_mu}. Suppose the sequence of step sizes satisfy $\sum_{k=1}^\infty \tau_k = \infty$, $\sum_{k=1}^\infty \tau_k^2 < \infty$, and $\sum_{k=1}^\infty \allowbreak D(\nu^{(k)} \|\nu^{(k-1)})  < \infty$, then
    \begin{equation}
        \lim_{k\to \infty} \int_{\RR^n} \| \nabla V_\lambda^{(k)} (x) \|_2^2 \ \mu^{(k)}(\dx) = 0.
    \end{equation}
\end{theorem}


\section{Numerical Results}
We first compute the \ac{cc} function for MIMO-AWGN channels with unit variance noise and cost constraint given by $\int \|x\|_2^2\  \mu(\dx) \le P$. Algorithm \ref{alg:compute} runs for $10\log_{10} P$ from $-10$ dB to $20$ dB, number of antennas $\{1\times1, 2\times2, 16\times16\}$, and number of particles $N\in\{64, 128, 256\}$ for $\{1\times1, 2\times2\}$ and $N\in\{64, 128\}$ for $16\times 16$. For the number of antennas greater than $1$, the channel matrices are initialized such that its square of eigenvalues sum to $1$. Dual ascent steps are also applied, so the resulting input power approaches the required level. The comparison between computation and theoretical results is illustrated in Fig. \ref{fig:mimoawgn_error}, and the resulting particle distributions for $1\times 1$ and $2\times 2$ cases in Fig. \ref{fig:mimo_particles}. It shows that the proposed method can compute comparable results to the theoretical values even for high-dimensional cases.

\begin{figure}
    \centering
        \includegraphics{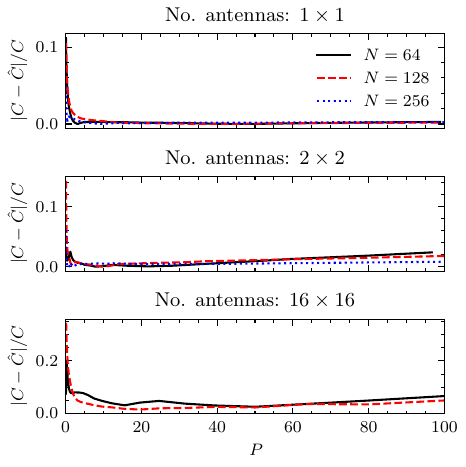}
        \caption{Error between the capacity $C$ and the computed values $\hat{C}$ vs. input power $P$ for different number of particles.}
    \label{fig:mimoawgn_error}
\end{figure}

\begin{figure}
    \centering
    \begin{subfigure}[c]{\linewidth}
    \centering
        \includegraphics{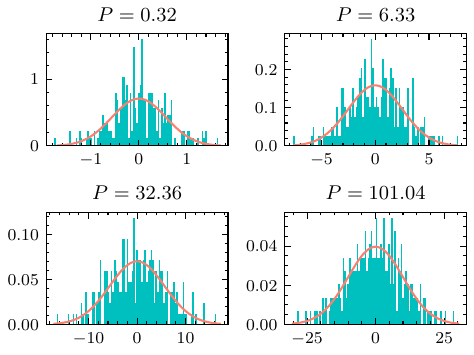}
        \caption{Single antenna: histograms of particles.}
    \end{subfigure}
    \centering
    \begin{subfigure}[c]{\linewidth}
    \centering
        \includegraphics{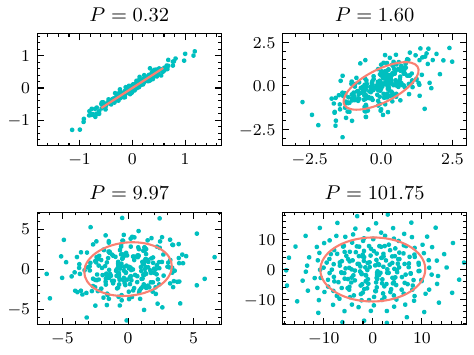}
        \caption{$2\times2$ antennas: locations of particles.}
    \end{subfigure}
    \caption{Resulting particles (blue) vs. theoretical optimal input distributions (red).}
    \label{fig:mimo_particles}
\end{figure}

The following example computes the capacity of the SISO Rayleigh fading channel $y=sx+z$ where $s$ and $z$ both follow Gaussian distribution $\calN(0,1)$, and the cost is the input power. We consider two cases: with and without \ac{csir}. For the case with \ac{csir}, the channel capacity is known as $\frac{1}{2}\mathbb{E}[\log(1+P|s|^2)]$\cite{goldsmith1997capacity} achieved by the Gaussian input distribution $\calN(0,P)$, where the expectation is taken with respect to $s$. To perform the proposed method, the \ac{csir} can be viewed as another channel output so that the equivalent channel distribution is given by $p_{YS|X}(y,s|x) = p_{Y|XS}(y|x,s)p_{S}(s)$. For the case without \ac{csir}, there is still no explicit form for the channel capacity. It is proven that the optimal input distribution is discrete with a finite number of mass points, one of which is always located at zero \cite{abou2001capacity}. The effective channel distribution is obtained as $p_{Y|X}(y|x) = \int p_{Y|XS}(y|x,s)p_{S}(s) \mathrm{d}s$ in this case, where the \ac{is} method may also be adopted. Fig. \ref{fig:fading_capacity} and \ref{fig:fading_particles} show the computation results as well as the particle values changing with steps, where the lower bound refers to the data rate computed by Gaussian input signaling for no \ac{csir} case. These outcomes again demonstrate the highly accurate performance of the proposed method.

\begin{figure}
    \centering
    \includegraphics{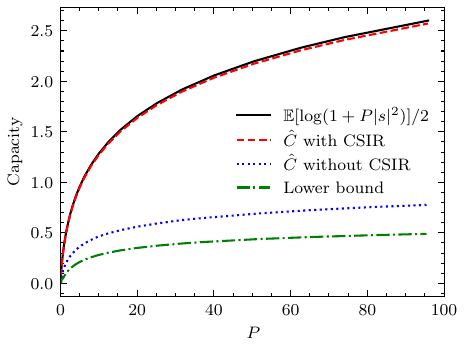}
    \caption{Results for fading channel with and without \ac{csir}.}
    \label{fig:fading_capacity}
\end{figure}

\begin{figure}
    \centering
    \includegraphics{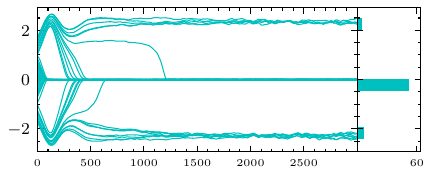}
    \caption{Evolution of particles over update steps and final histogram for fading channel without \ac{csir}, $P=1$.}
    \label{fig:fading_particles}
\end{figure}

\section{Conclusion}
This work proposes a numerical method to compute the \ac{cc} function for continuous channels based on Wasserstein gradient descent and \ac{is}. The convergence to a stationary point in the Wasserstein space is proved, and simulation results further validate the algorithm's effectiveness. 

\bibliographystyle{IEEEtran}
\bibliography{IEEEabrv,mybib.bib}
\end{document}